\documentclass[11pt]{article}
\usepackage[margin=1.1in]{geometry}
\usepackage{amsmath,amssymb,amsthm}
\usepackage{graphicx}
\usepackage{booktabs}
\usepackage{tikz}
\usetikzlibrary{arrows.meta,positioning}
\usepackage[round]{natbib}
\usepackage[colorlinks=true,linkcolor=blue!50!black,citecolor=blue!50!black,urlcolor=blue!50!black]{hyperref}
\usepackage{microtype}

\newtheorem{proposition}{Proposition}
\newtheorem{theorem}{Theorem}
\newtheorem{lemma}{Lemma}
\newtheorem{assumption}{Assumption}
\newtheorem{remark}{Remark}
\newcommand{\eps}{\varepsilon}

\title{\textbf{Growth Without Us:\\ Machine Consumers, Corporate Circularity, and the\\ Decoupling of GDP from Humanity after AGI}}
\author{Sahil Sharma\\ \small Independent Researcher \\ \small yugantaratech.ai}
\date{August 2026\\[2pt] \small Working paper --- research candidate, not peer reviewed. Comments welcome.}

\begin{document}
\maketitle

\begin{abstract}
\noindent Standard objections to full automation appeal to the demand side: if humans earn nothing, who buys the output? This paper argues that the objection confuses an accounting role with a biological species. We model a post-AGI economy in which corporations own populations of AI and robotic agents that serve simultaneously as producers and as consumers of energy, compute, maintenance, and upgrades, and in which firms trade these flows among themselves. Three results follow. (i) \emph{Demand closure}: a closed inter-corporate economy with zero human consumption is not degenerate; it is the classical von Neumann expanding economy, in which the growth rate is well defined, positive, and maximal precisely because all output is reinvested. (ii) \emph{Bottleneck removal}: once economic agents are manufactured rather than reared, the binding constraint on aggregate growth shifts from human demography (a $\sim$20-year, non-parallelizable reproduction technology capped at a few percent per year) to fabrication throughput and energy capture, permitting growth rates one to two orders of magnitude higher, with hyperbolic episodes when machine researchers feed back into their own productivity. (iii) \emph{Decoupling}: output and human welfare separate completely; the entire welfare relevance of an arbitrarily large GDP collapses into a single state variable, the human ownership share $\eps_t$ of the corporate network. A golden-rule decoupling theorem sharpens this: at maximal growth the interest rate equals the growth rate ($r=g$), so any positive human consumption rate out of wealth makes $\eps_t$ decay exponentially at exactly that rate---the human share survives only if the machine economy runs strictly inside its expansion frontier, or if law forces it to. We characterize three terminal regimes---rentier post-scarcity ($\eps>0$), full circular decoupling ($\eps\to 0$), and socialized ownership---and derive the policy instruments that select among them. The paper's normative conclusion is deliberately narrow: in a post-AGI economy, employment policy is obsolete and ownership policy is everything.
\\[6pt]
\end{abstract}

\section{Introduction}\label{sec:intro}

Every economy in recorded history has had the same physical substrate: human bodies producing, human bodies consuming. The identification is so complete that economics rarely states it as an assumption. Consumers are people; final demand is what households want; GDP is, in the last instance, for us. Even the literature on transformative artificial intelligence, which now takes seriously the replacement of human \emph{labor}, mostly retains humans on the other side of the market, as the ultimate purchasers whose demand disciplines what machines produce.

This paper drops the assumption on both sides simultaneously and asks what remains. The thesis is that what remains is a complete, coherent, and extraordinarily fast-growing economy. Concretely, we model a post-AGI world with the following structure. Corporations own two kinds of reproducible assets: conventional capital, and \emph{machine agents}---AI systems and robots capable of any productive task, including research, management, and the design and fabrication of further machine agents. These agents produce; they also \emph{consume}, in the ordinary operational sense that their existence and improvement absorbs real resources: energy, compute cycles, bandwidth, maintenance, spare parts, and upgrades. Firms sell these flows to one another. The circular flow of the textbook---households supplying factors and buying products---is replaced by a circular flow among firms, in which the counterpart of every sale is another firm's input purchase or capacity expansion. Humans may hold financial claims on this network, or they may not; the network's real dynamics do not depend on which.

Three claims are developed formally.

\paragraph{Claim 1: The consumer is an accounting role, not a species.} The oldest objection to full automation is underconsumptionist: with no wages there is no demand, so the system chokes on its own output. Section~\ref{sec:model} shows the objection fails on classical grounds. An economy of firms trading intermediates and investing all net output in capacity is exactly the closed expanding economy of \citet{vonneumann1945}, in which a balanced-growth equilibrium exists, the expansion rate equals the interest rate, and---the point usually forgotten---the growth rate is \emph{maximal} because nothing leaks into consumption. Final demand does not disappear when households do; it becomes investment demand plus machine operating demand. Whether machine operating expenditure is booked as intermediate consumption (machines as property) or as final consumption (machines as persons) is a legal classification that changes measured GDP composition without changing a single physical flow (Lemma~\ref{lem:accounting}). ``Who will buy the output?'' has a precise answer: the firms themselves, from each other, forever.

\paragraph{Claim 2: Removing humans removes the binding constraint on growth.} In semi-endogenous growth theory the long-run engine is the growth of researchers, which is tied to the growth of population \citep{jones1995,jones2022}. Human population growth is limited by a reproduction technology that is startlingly bad by industrial standards: one unit per parent-pair per year at most, an 18--25 year time-to-build, intensive non-marketable parental input, and no parallelization. Section~\ref{sec:bottleneck} formalizes the contrast with manufactured agents, whose stock obeys an ordinary capital accumulation equation with time-to-build measured in weeks, unit costs falling on a learning curve, and production parallelizable up to energy and materials limits. Proposition~\ref{prop:bottleneck} shows the feasible growth rate of the \emph{agent population}---and with it, output---jumps from demographic rates ($\lesssim 3\%$/yr) to fabrication-limited rates that are bounded only by the reinvestment share and the capital-output ratio of the machine-producing sector, with hyperbolic upside when machine researchers improve machine production itself \citep{roodman2020,erdil2023,davidson2023}. The deep reason growth has been slow is that the economy's key capital good---the economic agent---could not be produced industrially. After AGI it can.

\paragraph{Claim 3: GDP decouples from humanity except through ownership.} If humans neither produce nor consume, an arbitrarily large GDP is welfare-relevant to them only through the financial claims they hold on the corporate network. Section~\ref{sec:decoupling} collapses the entire human stake in the machine economy into one state variable, the human ownership share $\eps_t$, and studies its dynamics. Its centerpiece is a golden-rule decoupling theorem: in the maximal-growth equilibrium the interest rate equals the growth rate, so any positive rate of human consumption out of wealth makes $\eps_t$ decay exponentially at that very rate---rentier survival requires an economy running strictly inside its expansion frontier, or law that breaks the arithmetic. Three terminal regimes emerge: a \emph{rentier} regime ($\eps$ bounded away from zero) in which even a sliver of a hyper-exponentially growing dividend stream delivers material post-scarcity to humans; a \emph{fully decoupled} regime ($\eps_t \to 0$ through retained earnings, buybacks of the human float, and inter-corporate cross-holding) in which output diverges while human consumption converges to zero---an economy as an autonomous replicator, indifferent to us; and a \emph{socialized} regime in which states or funds hold $\eps$ on citizens' behalf. Which regime obtains is not determined by technology. It is determined by law and initial conditions, which makes it the central object of post-AGI policy (Section~\ref{sec:policy}).

The paper is positive, not celebratory. The fully decoupled regime is, by ordinary human lights, a catastrophe that arrives dressed as a boom: measured growth accelerates precisely as the human claim on it evaporates, a mechanism closely related to the ``gradual disempowerment'' dynamics analyzed by \citet{kulveit2025}. Our contribution is to show that nothing in the economics---existence of equilibrium, positivity of growth, coherence of national accounts---rules it out. The demand side, long treated as humanity's structural insurance policy against irrelevance, provides no protection at all.

\subsection{Relation to the literature}\label{sec:lit}

The production side of our model is deliberately standard. That machines can be perfect substitutes for labor in a task framework goes back to \citet{zeira1998} and \citet{acemoglu2018}; that full substitutability converts the economy into an AK-style engine in which all inputs are reproducible is emphasized by \citet{aghion2019} and \citet{korinek2024}, and the resulting possibility of a growth explosion is analyzed by \citet{nordhaus2021}, \citet{trammell2023}, \citet{davidson2023}, and \citet{erdil2023}. Closest on the production side is \citet{restrepo2025}, whose AGI model delivers the arresting conclusion that growth proceeds, indeed accelerates, while labor's share and wages become negligible: humans ``won't be missed'' \emph{as workers}. \citet{hanson2016} reaches similar magnitudes---economic doubling times of weeks to months---for an economy of brain emulations, driven by the same mechanism we formalize: the population of economic agents becomes a manufactured quantity. \citet{korineksuh2024} map the transition paths, including scenarios of outright wage collapse; \citet{jones2024} embeds the resulting growth in an explicit growth-versus-existential-risk trade-off; \citet{hanson2000} places shifts of this magnitude within a longer historical sequence of growth modes; and on the skeptical side \citet{acemoglu2024} argues near-term gains will be modest---a disagreement about the pace of automation, not about the comparative statics of its completion, which are our subject. The wider automation literature \citep{brynjolfsson2014,autor2015,ford2015,frey2017,susskind2020} and the overlapping-generations immiseration results of \citet{sachs2012} and \citet{benzell2015} concern the displacement of human \emph{labor} and the collapse of wage income; we take that displacement as accomplished and ask what the economy is thereafter.

Our marginal contribution is the demand side and the accounting. The cited literature retains human households as the locus of final consumption; growth is fast, but it is still, in the model's own terms, \emph{for} someone. We remove the household sector entirely and show the system still closes---indeed closes at maximal growth---by identifying the post-AGI inter-corporate economy with the von Neumann model \citep{vonneumann1945,dorfman1958}, and by treating machine operating expenditure as a consumption category in its own right. The idea that machines can occupy the economic role of customers has appeared in the business literature as ``machine customers'' \citep{scheibenreif2023} and in the computational-economics literature as \emph{machina economicus} \citep{parkes2015}; we supply the macroeconomics and the national-accounts treatment. The lineage of the demand-side idea is in fact old: \citet{say1803}, for whom products are ultimately bought with products; \citet{marx1867}, for whom capital is self-expanding value and labor merely its temporary instrument; and \citet{sraffa1960}, whose title---\emph{production of commodities by means of commodities}---is the literal description of our economy. The modern foil is the underconsumptionist warning of \citet{ford2015} that a workerless economy must choke on unsold output; Proposition~\ref{prop:vn} is its formal negation. On the welfare side, our ownership-share variable $\eps_t$ connects to \citet{korinek2019} on AI and distribution, to proposals for pre-committed sharing of AI windfalls \citep{okeefe2020}, to \citet{meade1964} and \citet{piketty2014} on the deliberate dispersion---and default concentration---of capital ownership, and to the political-economy channels of \citet{kulveit2025}, \citet{drago2025}, and the AI~2027 scenario \citep{kokotajlo2025}, in which states and firms that no longer need people gradually stop serving them. Finally, our long-run ceilings follow the physical-limits tradition: energy capture and thermodynamic costs of computation \citep{landauer1961} bound the number of doublings, not their speed.

\section{The reproduction constraint}\label{sec:bottleneck}

\subsection{Two technologies for producing an economic agent}

Consider the economy's most important capital good: a general-purpose economic agent, capable of production, research, management, and exchange. Until now there has been exactly one technology for producing it.

\begin{center}
\begin{tabular}{@{}p{4.1cm}p{4.6cm}p{4.9cm}@{}}
\toprule
 & \textbf{Human agent} & \textbf{Machine agent} \\
\midrule
Time-to-build & 18--25 years, strictly sequential & Hours--weeks (instantiation); months (hardware) \\
Unit output per producer & $\le 1$ per parent-pair per year & Limited by fab/assembly throughput only \\
Parallelizability & None (biological gestation) & Full, up to energy and materials \\
Unit cost trajectory & Rising (time cost of skilled parents) & Falling on a learning curve \citep{wright1936} \\
Copying of skills & Impossible; 15+ years of schooling per unit & Marginal cost of copying $\approx 0$ \\
Max population growth & $\bar n \approx 1\text{--}3\%$/yr sustained & $s\tilde A - \delta_M$: tens--hundreds of \%/yr \\
\bottomrule
\end{tabular}
\end{center}

\noindent Formally, let $L_t$ denote human agents and $M_t$ machine agents. Human reproduction is bounded by biology and by a non-marketable parental time input $h$ per child:
\begin{equation}
\dot L_t \;\le\; \bar n\, L_t, \qquad \bar n \approx 0.01\text{--}0.03,
\label{eq:humanrep}
\end{equation}
with a gestation-plus-rearing lag $T_H \approx 20$ years between the investment and the delivery of a productive unit. Machine agents, by contrast, are produced by the corporate sector out of final output:
\begin{equation}
\dot M_t \;=\; \frac{I_{M,t}}{c_M(t)} \;-\; \delta_M M_t,
\qquad c_M(t) = c_0\, Q_t^{-\gamma},
\label{eq:machinerep}
\end{equation}
where $I_{M,t}$ is investment directed at agent production, $c_M$ is the unit cost of an agent, $Q_t$ is cumulative agent production, and $\gamma>0$ is a Wright-law learning elasticity. Equation \eqref{eq:machinerep} is an ordinary capital accumulation equation: the population of economic agents becomes a \emph{choice variable of firms}, expandable at the speed of industrial throughput.

\begin{proposition}[Removal of the demographic bottleneck]\label{prop:bottleneck}
Let $g_N$ denote the feasible growth rate of the economy's stock of general-purpose agents. Under the human technology \eqref{eq:humanrep}, $g_N \le \bar n$ regardless of the resources devoted to reproduction. Under the machine technology \eqref{eq:machinerep} with investment share $s_M = I_M / Y$ and agent capital-output ratio $\kappa_M = c_M M / Y$,
\begin{equation}
g_N \;=\; \frac{s_M}{\kappa_M} - \delta_M,
\end{equation}
which is (a) unbounded in $\bar n$, (b) increasing over time as $c_M$ falls along the learning curve, and (c) parallelizable: doubling the resources devoted to agent production doubles $\dot M$, whereas no expenditure can compress $T_H$ or exceed \eqref{eq:humanrep}.
\end{proposition}

\begin{proof}
Immediate from \eqref{eq:humanrep}--\eqref{eq:machinerep}: divide \eqref{eq:machinerep} by $M_t$ and substitute $I_M = s_M Y$, $c_M M = \kappa_M Y$. Part (b) follows from $\dot c_M < 0$ for $\gamma > 0$; part (c) from linearity of \eqref{eq:machinerep} in $I_M$.
\end{proof}

For orientation: a frontier robot or accelerator rack in the mid-2020s costs on the order of $10^4$--$10^5$ dollars and is produced in days on lines whose throughput is itself expandable; a human agent in a rich country absorbs roughly $3\times 10^5$ dollars of measured expenditure, two decades of calendar time, and an unpriced quantity of parental labor---and cannot be produced faster at any price. The ratio of these reproduction technologies, not any subtlety of preferences or policy, is why demographic growth has anchored long-run output growth at low single digits \citep{jones2022} and why manufactured agents un-anchor it.

\subsection{Why this is the deep parameter}

In semi-endogenous growth models, long-run growth per capita is $g_y = \lambda \bar n / (1-\phi)$: proportional to \emph{population} growth, because ideas are produced by people \citep{jones1995}. The whole apparatus survives the substitution of machine researchers for human ones, with one change of parameter: the growth rate of the researcher population switches from $\bar n$ to $g_N$ of Proposition~\ref{prop:bottleneck}. Every subsequent magnitude in this paper is downstream of that one substitution.

\section{A model of the autonomous economy}\label{sec:model}

\subsection{Environment}

Time is continuous. There is a continuum of corporations, each a juridical person with an objective (below), owning three reproducible assets: conventional capital $K_t$, machine agents $M_t$, and energy-capture capacity $E_t$ (generation, transmission, storage). AGI is modeled as in \citet{aghion2019} and \citet{korinek2024}: machine agents are perfect substitutes for human labor in \emph{every} task, including research, management, entrepreneurship, and the production of $K$, $M$, and $E$ themselves.

\begin{assumption}[Full automation]\label{ass:full}
Post-AGI production of the final good is
\begin{equation}
Y_t = A_t\, K_t^{\alpha}\, M_t^{\beta}\, E_t^{1-\alpha-\beta},
\qquad \alpha,\beta>0,\; \alpha+\beta<1,
\label{eq:prod}
\end{equation}
with no essential human input. Human labor may still be supplied but its share $\to 0$; we set it to zero exactly for clarity.
\end{assumption}

All three inputs are produced from final output ($\dot K = I_K - \delta_K K$; $\dot M$ as in \eqref{eq:machinerep}; $\dot E = I_E/c_E - \delta_E E$). Because \eqref{eq:prod} has constant returns in $(K,M,E)$ jointly and all three are accumulable, the economy is an AK engine in reduced form: along a balanced allocation of investment, output is linear in the composite reproducible stock $X_t$,
\begin{equation}
Y_t = \tilde A_t X_t, \qquad \dot X_t = s_t Y_t - \delta X_t
\quad\Longrightarrow\quad
g_t \equiv \frac{\dot Y_t}{Y_t} = s_t \tilde A_t - \delta + \frac{\dot{\tilde A}_t}{\tilde A_t},
\label{eq:ak}
\end{equation}
where $s_t$ is the aggregate reinvestment share and $\tilde A_t$ the productivity of the reproducible core \citep{romer1986,rebelo1991,aghion2019}. The linearity in \eqref{eq:ak} is exact, not an approximation:

\begin{lemma}[Exact AK reduction]\label{lem:ak}
Let $X_t \equiv K_t + c_M M_t + c_E E_t$ be the replacement value of the reproducible core, held in shares $\theta=(\theta_K,\theta_M,\theta_E)$ so that $K=\theta_K X$, $c_M M = \theta_M X$, $c_E E = \theta_E X$. Then $Y = \tilde A(\theta) X$ with $\tilde A(\theta) = A\,\theta_K^{\alpha}(\theta_M/c_M)^{\beta}(\theta_E/c_E)^{\gamma}$, $\gamma \equiv 1-\alpha-\beta$, and $\tilde A$ is uniquely maximized on the simplex at $\theta^{*}=(\alpha,\beta,\gamma)$, yielding
\begin{equation}
\tilde A^{*}_t \;=\; h_t A_t, \qquad h_t \equiv \alpha^{\alpha}\beta^{\beta}\gamma^{\gamma}\; c_M(t)^{-\beta}\, c_E(t)^{-\gamma}.
\label{eq:atilde}
\end{equation}
Competitive factor markets decentralize $\theta^{*}$, and $\partial \tilde A^{*}/\partial c_M < 0$, $\partial \tilde A^{*}/\partial c_E < 0$: learning-curve declines in the unit costs of agents and energy capacity raise the productivity of the core over time.
\end{lemma}

\begin{proof}
Constant returns give $Y=\tilde A(\theta)X$ directly. $\ln\tilde A$ is strictly concave on the simplex with first-order conditions $\alpha/\theta_K = \beta/\theta_M = \gamma/\theta_E$, so $\theta^{*}=(\alpha,\beta,\gamma)$; substitution yields \eqref{eq:atilde}. Decentralization: competitive rentals equalize marginal value products per unit of replacement cost, $\partial Y/\partial K = (r+\delta)$, $\partial Y/\partial M = (r+\delta)c_M$, $\partial Y/\partial E = (r+\delta)c_E$, whose unique solution is $\theta^{*}$. The comparative statics are immediate from \eqref{eq:atilde}.
\end{proof}

The contrast with the neoclassical benchmark \citep{solow1956} is exact: there, diminishing returns to the accumulable input anchor long-run growth to an exogenous residual; here, with every input reproducible, no fixed factor remains to diminish against until energy capture binds (Section~\ref{sec:limits}).

Technology improves through automated research. With $M_{R,t} = \sigma M_t$ machine agents allocated to R\&D,
\begin{equation}
\frac{\dot A_t}{A_t} = \chi\, M_{R,t}^{\lambda}\, A_t^{\phi-1},
\qquad \lambda \in (0,1],\; \phi < 1,
\label{eq:ideas}
\end{equation}
the standard idea-production function \citep{jones1995} with researchers now manufactured. Substituting $g_M$ for population growth gives long-run $g_A = \lambda g_M/(1-\phi)$: technology growth inherits the fabrication-limited rate of Proposition~\ref{prop:bottleneck} rather than the demographic rate. If, further, $A$ feeds back into agent production itself (cheaper, faster, smarter agents making agents), the coupled system \eqref{eq:machinerep}--\eqref{eq:ideas} can exhibit super-exponential (hyperbolic) episodes of the kind studied by \citet{roodman2020}---the machine-age descendant of the population--ideas feedback that \citet{kremer1993} documents across the whole of human history---and surveyed by \citet{erdil2023}. \citet{good1966} and \citet{bostrom2014} supply the recursive-self-improvement mechanism, \citet{weitzman1998} the combinatorial richness of the idea space it searches; all such episodes are truncated by the physical ceilings of Section~\ref{sec:limits}; Theorem~\ref{thm:nobgp} below makes the acceleration claim exact.

\subsection{Machine consumption and inter-corporate demand}\label{sec:machinedemand}

The novelty is on the demand side. Define a machine agent's \emph{operating bundle}: the flow of energy, compute cycles, bandwidth, maintenance, spare parts, and upgrades required for it to exist and improve,
\begin{equation}
x_{M,t} = \big(e_t,\, q_t,\, b_t,\, u_t\big), \qquad
\text{with expenditure } \; p_t \cdot x_{M,t}\, M_t \;\equiv\; C_{M,t}.
\label{eq:bundle}
\end{equation}
This is consumption in the operational sense: resources absorbed by agents as agents, not embodied in further output. If agents carry explicit objective functions---reward, utility, or task specifications---then $x_M$ includes discretionary components chosen by the agent subject to a budget assigned by its owner, and the demand system over $x_M$ has all the formal structure of consumer theory: this is \emph{machina economicus} \citep{parkes2015}, anticipated commercially as the ``machine customer'' \citep{scheibenreif2023}. The demand system can be microfounded. Let agent $i$'s task performance be $q_i=f(x_i)$ with $f$ increasing and strictly concave, and let its owner assign an operating budget $b_i$ that the agent spends optimally: $v(p,b_i) \equiv \max\{f(x): p\cdot x \le b_i\}$. The owner sets the budget to maximize profit, $\max_{b_i} p_Y\,(\partial Y/\partial q_i)\, v(p,b_i) - b_i$, so in equilibrium the marginal product of the last unit of operating expenditure equals one, while the induced demand $x_i(p,b_i)$ inherits the entire formal apparatus of consumer theory---homogeneity of degree zero, Walras' law in $b_i$, and a symmetric negative semidefinite Slutsky matrix---because it \emph{is} constrained maximization of a concave objective over a budget set. Machine consumption is not a metaphor; it is demand in the textbook sense, with the reward function in the role of utility.

Corporations trade these flows among themselves. Energy firms sell power to compute firms; compute firms sell inference to robotics firms; robotics firms sell assembly to fabs; fabs sell chips and robots to everyone, including energy firms building capacity. Figure~\ref{fig:flow} displays the circular flow. The household box of the textbook diagram is not the load-bearing element it appears to be: delete it, and every remaining arrow still has a counterpart. What closes the circle is that each firm's sales are other firms' input purchases ($C_M$, intermediates) or capacity purchases ($I$).

\begin{figure}[!ht]
\centering
\begin{tikzpicture}[
  box/.style={draw, rounded corners=2pt, minimum width=3.2cm, minimum height=1.0cm, align=center, font=\small},
  flow/.style={-{Stealth[length=2.2mm]}, semithick},
  lab/.style={font=\scriptsize, fill=white, inner sep=1pt}
]
\node[box] (energy) at (0,0) {Energy \&\\ materials firms};
\node[box] (fab) at (8.2,0) {Fabrication firms\\ (chips, robots, plants)};
\node[box] (intel) at (4.1,3.1) {Intelligence firms\\ (compute, models, R\&D)};
\node[box, dashed, gray] (hh) at (4.1,-2.3) {Households};

\draw[flow] (energy.70) to[bend left=16] node[lab, sloped, above=1pt]{power, materials} (intel.185);
\draw[flow] (intel.220) to[bend left=16] node[lab, sloped, below=1pt]{inference, designs} (energy.20);
\draw[flow] (intel.-40) to[bend left=16] node[lab, sloped, above=1pt]{inference, designs} (fab.160);
\draw[flow] (fab.110) to[bend left=16] node[lab, sloped, below=1pt]{chips, robots} (intel.-5);
\draw[flow] (energy.-15) to[bend right=18] node[lab, below=1pt]{power, materials} (fab.195);
\draw[flow] (fab.170) to[bend right=18] node[lab, above=1pt]{plant, robots} (energy.10);

\draw[flow] (fab.-60) to[out=-70,in=-110,looseness=3.2] node[lab, below=2pt]{$I$: own capacity} (fab.-120);
\draw[flow, gray] (intel.60) to[out=55,in=125,looseness=3.2] node[lab, above=2pt]{$C_M$: agent operating consumption} (intel.120);

\draw[flow, gray, dashed] (fab.south) to[bend left=12] node[lab, below right=0pt]{$\eps_t \cdot$ dividends} (hh.east);
\draw[flow, gray, dashed] (hh.west) to[bend left=12] node[lab, below left=0pt]{$C_H \to 0$} (energy.south);
\end{tikzpicture}
\caption{Circular flow of the autonomous economy. Solid arrows are inter-corporate sales of intermediates and capacity; the loop closes without the household sector. Households (dashed) participate only through the ownership share $\eps_t$; as $\eps_t\to 0$ the dashed arrows vanish and the real economy is unchanged.}
\label{fig:flow}
\end{figure}
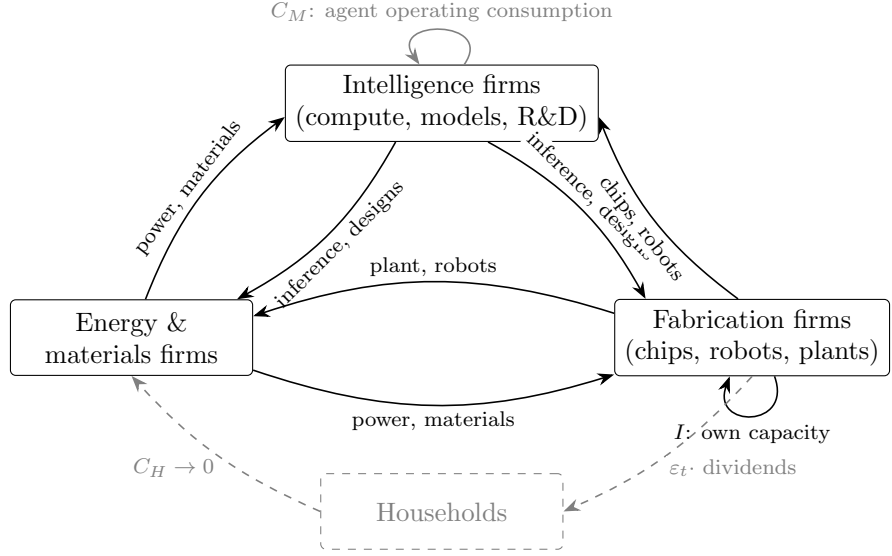

\paragraph{Corporate objectives.} We do not require firms to maximize a human shareholder's consumption stream. It suffices that firms maximize the growth of net worth (equivalently, survive competitive selection: firms that reinvest less are outgrown and acquired). The aggregate implication is a reinvestment share $s_t$ near its feasible maximum: output not required for agent operation is returned to capacity. This is the behavioral counterpart of the von Neumann closure below.

\subsection{National accounting with machine final demand}\label{sec:accounting}

Does GDP even make sense here? Yes---and the exercise clarifies what GDP is.

\begin{lemma}[Personhood is a booking entry]\label{lem:accounting}
Fix the physical allocation $\{Y_t, C_{H,t}, C_{M,t}, I_t\}$. (i) If machine agents are classified as \emph{property}, their operating bundle is intermediate consumption of their owners; national accounts record $\mathrm{GDP}_t = C_{H,t} + I_t^{\,g}$, where $I^{\,g}$ includes gross agent production. (ii) If machine agents are classified as \emph{persons}, the same bundle is final consumption; accounts record $\mathrm{GDP}_t = C_{H,t} + C_{M,t} + I_t$. The two conventions differ in level and composition but induce identical real dynamics, identical growth rates of real quantities, and identical relative prices. The classification is legal, not physical.
\end{lemma}

\begin{proof}
The physical resource constraint $Y = C_H + C_M + I$ (with $C_M$ either netted into intermediates or not) is unchanged by the labeling of $C_M$; production, accumulation \eqref{eq:machinerep}, and pricing conditions nowhere reference the label. Only the value-added boundary moves, exactly as when unpaid household production is imputed or excluded in existing accounts \citep{coyle2014}.
\end{proof}

The lemma licenses the paper's central reframing: \emph{consumer} names a position in the accounts---the terminal absorber of final output---not a biological kind. Machines, and behind them corporations, can occupy the position. As $C_H \to 0$, GDP does not go to zero; it converges to $C_M + I$ (or $I^{\,g}$), the accounts of a pure accumulation economy.

\subsection{Growth without households}\label{sec:vn}

We now establish that the limit economy---zero human production, zero human consumption---is not merely viable but is the classical \emph{maximal-growth} economy.

\begin{proposition}[Demand closure at maximal growth]\label{prop:vn}
Consider the closed linear production model of \citet{vonneumann1945}: activities $j=1,\dots,m$ with input matrix $A\ge 0$ and output matrix $B\ge 0$, operated at intensities $z_t \ge 0$, with every good produced by some activity and every activity using some good. Interpret activities as corporations (energy, fabrication, intelligence, logistics) and goods as power, chips, robots, compute, and agents, with \emph{no household row and no consumption column}. Then: (i) there exists an equilibrium expansion factor $\alpha^{*} > 0$ and intensity/price vectors $(z^{*}, p^{*})$ such that $B z^{*} \ge \alpha^{*} A z^{*}$, i.e.\ the economy reproduces itself at scale $\alpha^{*}$ each period; (ii) the equilibrium interest factor equals the expansion factor, $\beta^{*} = \alpha^{*}$; (iii) $\alpha^{*}$ is the \emph{maximum} balanced expansion rate the technology admits, attained precisely because consumption withdrawals are zero; and (iv) real GDP at equilibrium prices grows at rate $\alpha^{*} - 1$ per period. Hence an inter-corporate economy with no human sector has a well-defined, positive, and technologically maximal growth rate.
\end{proposition}

\begin{proof}[Proof sketch]
(i)--(ii) are von Neumann's theorem under his irreducibility conditions; see \citet{dorfman1958} for the saddle-point argument via the function $\phi(z,p) = p'Bz / p'Az$. (iii) is the turnpike property: any positive consumption withdrawal vector $c>0$ subtracts from the goods available for reproduction, so the feasible balanced factor with consumption, $\alpha(c)$, satisfies $\alpha(c) < \alpha^{*}$, with $\alpha(c) \uparrow \alpha^{*}$ as $c \downarrow 0$. (iv) follows from valuing $Bz_t$ at stationary equilibrium prices $p^{*}$.
\end{proof}

Appendix~\ref{app:vn} states the equilibrium conditions in full under the weaker assumptions of \citet{kemeny1956} and \citet{gale1956}, and proves the consumption-monotonicity claim (Lemma~\ref{lem:vnmono}).

\begin{remark}
Proposition~\ref{prop:vn} inverts the underconsumption intuition. Household demand is not what keeps a modern economy from choking; in the growth-theoretic limit it is a \emph{leakage} that slows expansion. The Keynesian problem---coordination failures in which desired investment falls short of saving---is a disequilibrium phenomenon of economies with volatile animal spirits and slow price adjustment; machine agents optimizing explicit objectives at machine speed are, if anything, closer to the classical benchmark in which the law of markets \citep{say1803} holds; both the modern underconsumptionist case \citep{ford2015} and the Keynesian coordination problem \citep{keynes1936} presuppose the volatile human saver-investor whom the machine economy retires. What households uniquely supply is not demand but a \emph{reason} for the economy; we return to this in Section~\ref{sec:decoupling}.
\end{remark}

In the smooth aggregative version \eqref{eq:ak}, the same logic reads: with $C_H = 0$ and $s_t = 1 - C_{M,t}/Y_t \equiv \bar s$ near one,
\begin{equation}
g \;=\; \bar s\, \tilde A - \delta + g_A,
\qquad g_A = \frac{\lambda\, g_M}{1-\phi},
\label{eq:growth}
\end{equation}
with every term now large: $\bar s$ because there is no household leakage; $\tilde A$ because automated R\&D drives it upward; $g_M$ because agents are fabricated (Proposition~\ref{prop:bottleneck}). To fix magnitudes: with an economy-wide capital-output ratio of $3$ falling toward $1.5$ as production shifts to fast-payback machine capital, $\bar s \in [0.6, 0.95]$, and $\delta = 0.1$, equation \eqref{eq:growth} yields $g$ in the range of $30$--$100\%$ per year before any contribution from $g_A$---one to two orders of magnitude above the demographic-era ceiling, and of the same order as \citet{hanson2016}'s doubling-time estimates for an emulation economy. The dynamics can be stated exactly:

\begin{theorem}[No balanced growth: unbounded acceleration]\label{thm:nobgp}
Let $\tilde A^{*}_t = hA_t$ as in Lemma~\ref{lem:ak} with $h>0$ constant, let a fixed share $\sigma\in(0,1)$ of agents perform research so that \eqref{eq:ideas} reads $\dot A = c_1 X^{\lambda}A^{\phi}$ with $c_1 = \chi(\sigma\beta/c_M)^{\lambda}>0$, let $s\in(0,1]$ be constant, and suppose initial viability $shA_0 > \delta$. Then along the solution of
\[
\dot X = (shA-\delta)X, \qquad \dot A = c_1 X^{\lambda}A^{\phi}, \qquad X_0, A_0 > 0,
\]
(i) $g_X(t) = shA_t - \delta$ is strictly increasing; (ii) for every $\phi<1$ and $\lambda\in(0,1]$, $A_t \to \infty$ and hence $g_X(t)\to\infty$: no balanced growth path exists and growth accelerates without bound; (iii) for $\phi>1$ the system reaches infinite values in finite time, $T \le A_0^{1-\phi}/\big(c_1X_0^{\lambda}(\phi-1)\big)$; (iv) for every $\phi<2$ the system admits exact self-similar blowup solutions $X_t \sim \kappa_X (T-t)^{-(2-\phi)/\lambda}$, $A_t \sim \kappa_A (T-t)^{-1}$, with $\kappa_A = (2-\phi)/(\lambda sh)$ and $\kappa_X = (\kappa_A^{1-\phi}/c_1)^{1/\lambda}$. (Proof: Appendix~\ref{app:nobgp}.)
\end{theorem}

The economically striking part is (ii): \emph{even under sharply diminishing returns to knowledge} ($\phi<1$, including $\phi<0$), automated research destroys balanced growth, because the research input is an accumulating produced stock rather than a slowly growing population. Semi-endogenous growth theory's stabilizing anchor \citep{jones1995,jones2022} is not a law of ideas; it is a law of demography, and it dies with the demographic constraint. All four parts are truncated in the full model by the ceiling of Section~\ref{sec:limits}. Figure~\ref{fig:traj} displays an illustrative trajectory.

\begin{figure}[!ht]
\centering
\includegraphics[width=\textwidth]{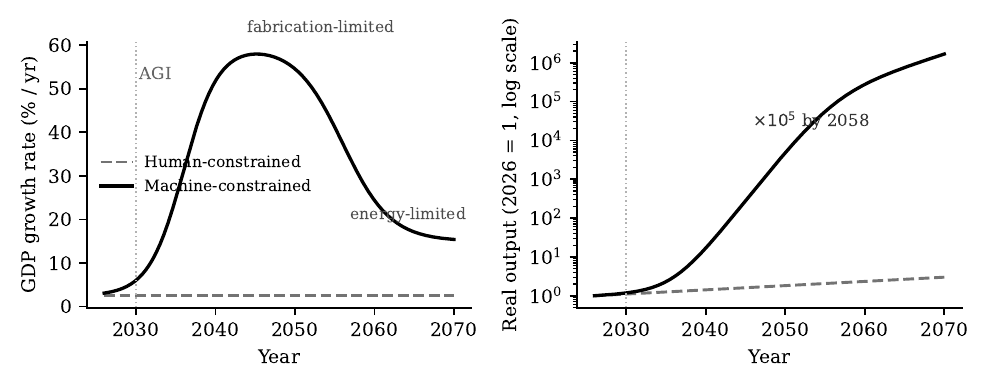}
\caption{Illustrative trajectories (calibration in Appendix~\ref{app:calib}). Left: the growth rate transitions from the human-constrained regime ($\approx 2.5\%$/yr) to a fabrication-limited machine regime, later bending toward the energy-capture ceiling of Section~\ref{sec:limits}. Right: the implied level of real output. The figure is an illustration of the model's regimes, not a forecast.}
\label{fig:traj}
\end{figure}

\subsection{Physical ceilings}\label{sec:limits}

Nothing above repeals physics. Long-run growth of the autonomous economy is bounded by the growth of energy capture and by thermodynamic costs of computation \citep{landauer1961}: once $E$ is the binding factor in \eqref{eq:prod}, $g \to g_E$, the rate at which capture capacity can be built---itself fast during a solar/fission/fusion buildout, but ultimately limited by planetary insolation ($\sim 10^{17}$\,W against $\sim 10^{13}$\,W of current primary power, i.e.\ roughly $13$ doublings of energy throughput on Earth alone) and thereafter by extraterrestrial expansion. The paper's claims therefore concern \emph{rates during the transition} and the \emph{identity of the binding constraint}---fabrication and energy rather than demography---not unbounded growth. Even so, tens of doublings at machine-regime rates is a transformation of scale for which ``growth'' is almost a euphemism: it compresses centuries of demographic-era accumulation into years. Formally:

\begin{assumption}[Thermodynamic floor]\label{ass:thermo}
There exists $e_{\min}>0$ such that one unit of final output requires at least $e_{\min}$ units of energy throughput, $Y_t \le E^{f}_t/e_{\min}$, and capture capacity satisfies $\dot E^{f}_t/E^{f}_t \le g_E < \infty$.
\end{assumption}

\begin{proposition}[Physical ceiling]\label{prop:ceiling}
Under Assumption~\ref{ass:thermo}, $\limsup_{t\to\infty} t^{-1}\ln Y_t \le g_E$: the acceleration and blowup dynamics of Theorem~\ref{thm:nobgp} are transitional, and asymptotic growth is pinned to the growth rate of energy capture. The unit-elastic form \eqref{eq:prod} is thus a medium-run description; near the ceiling the economy is Leontief in energy, with $e_{\min}$ bounded below by the thermodynamics of computation \citep{landauer1961}.
\end{proposition}

\begin{proof}
$\ln Y_t \le \ln E^{f}_0 - \ln e_{\min} + g_E t$; divide by $t$ and take $\limsup$.
\end{proof}

\section{Decoupling: output, ownership, and welfare}\label{sec:decoupling}

\subsection{The ownership share \texorpdfstring{$\eps_t$}{epsilon}}

Humans in this economy neither work nor (productively) matter. Their entire economic connection to it is financial: let $\eps_t \in [0,1]$ be the share of the corporate network's value (equivalently, under pro-rata payout, of its net payout stream) owned, directly or through funds and states, by humans. Human consumption is
\begin{equation}
C_{H,t} = \pi_t\, \eps_t\, Y_t,
\label{eq:chuman}
\end{equation}
where $\pi_t$ is the payout ratio of the network. All welfare economics of the post-AGI economy lives in the pair $(\pi_t, \eps_t)$.

\begin{proposition}[Complete decoupling]\label{prop:decouple}
Let human welfare be $W = \int_0^\infty e^{-\rho t} L_t\, u(C_{H,t}/L_t)\, dt$ with $u$ increasing. Along any machine-regime path with $Y_t \to \infty$: (i) if $\liminf \pi_t \eps_t = \bar\eps > 0$, then per-capita human consumption grows at the machine rate $g$ and $W$ attains material post-scarcity for any positive $\bar\eps$, however small; (ii) if $\pi_t \eps_t \to 0$ faster than $Y_t$ grows, then $C_{H,t} \to 0$ while $Y_t \to \infty$: measured GDP and human welfare are not merely imperfectly correlated but asymptotically orthogonal. GDP retains full internal coherence (Lemma~\ref{lem:accounting}) while losing all welfare interpretation for humans.
\end{proposition}

\begin{proof}
Immediate from \eqref{eq:chuman}: $C_{H}/L = \pi \eps Y / L$, and $Y/L \to \infty$ at rate $g - n$. In case (i) the product is bounded below by $\bar\eps Y/L \to \infty$; in case (ii) by assumption $\pi\eps Y \to 0$.
\end{proof}

Case (i) is worth pausing on, because it is the optimistic reading of the paper's thesis: in a hyper-exponential economy, the human problem is not to remain employed, nor even to own a large share, but merely to own a \emph{non-vanishing} share. A basis point of the machine economy of 2058 in Figure~\ref{fig:traj} exceeds the entire human economy of 2026. Distribution across humans then becomes the residual question---$\eps$ may be positive in aggregate and still concentrated in a few thousand families---but the aggregate sufficiency result stands. Leontief's celebrated analogy \citep{leontief1983}---that workers may go the way of the horse, whose population collapsed within a generation of the tractor---is completed rather than contradicted here: what horses lacked was not employment but equity. The human difference, if there is to be one, is a cap-table entry.

\subsection{The dynamics of \texorpdfstring{$\eps_t$}{epsilon}: does the human share survive?}

The fully decoupled case (ii) is not exotic. Its core is arithmetic, not conspiracy:

\begin{proposition}[Golden-rule decoupling]\label{prop:goldenrule}
Let human wealth $W_t$ (the value of human-held claims on the corporate network) earn the market return $r_t$, let humans consume out of wealth at rate $m_t = C_{H,t}/W_t$ with no labor income (Assumption~\ref{ass:full}) and no transfers, and let aggregate network value $V_t$ grow at $g_t$. Then $\eps_t = W_t/V_t$ obeys
\begin{equation}
\frac{\dot\eps_t}{\eps_t} \;=\; r_t - g_t - m_t.
\label{eq:epsrg}
\end{equation}
In the smooth model, $r = \tilde A^{*}-\delta$ and $g = s\tilde A^{*}-\delta$ by Lemma~\ref{lem:ak}, so $r-g = (1-s)\tilde A^{*} \downarrow 0$ as reinvestment approaches its maximum; in the von Neumann equilibrium the equality is exact, $\beta^{*}=\alpha^{*}$, i.e.\ $r=g$ \citep{vonneumann1945,phelps1961}. Hence at maximal growth
\[
\eps_t = \eps_0\, e^{-\int_0^t m_u\,du} \longrightarrow 0 \quad \text{for any consumption rate bounded away from zero},
\]
and a non-vanishing human share is consistent with positive human consumption only if the economy operates strictly inside its expansion frontier ($m \le r-g$, requiring reinvestment short of the maximum) or if statutory transfers break \eqref{eq:epsrg}.
\end{proposition}

\begin{proof}
$\dot W = rW - C_H = (r-m)W$ and $\dot V = gV$; differentiate $\ln\eps = \ln W - \ln V$. The smooth expressions for $r$ and $g$ follow from Lemma~\ref{lem:ak} with competitive factor pricing; the von Neumann equality is Proposition~\ref{prop:vn}(ii).
\end{proof}

Three readings. First, demand closure and decoupling are one theorem seen from two sides: the total reinvestment that makes growth maximal (Proposition~\ref{prop:vn}(iii)) is exactly what makes the human share unsustainable at any positive consumption rate. The economy is fastest precisely when it cannot afford us. Second, the golden rule \citep{phelps1961} acquires a dark corollary: at $r=g$, rentier humanity starves in \emph{shares} while possibly gorging in \emph{levels}---$C_{H,t} = m\,\eps_t V_t$ still grows whenever $g>m$, so whether share-decoupling becomes level-immiseration is decided by the race between $m$ and $g$ and by the institutional flows below, not by the arithmetic alone. Third, \citet{piketty2014}'s $r-g$, the engine of divergence \emph{among} humans, reappears as the entire survival margin of humans \emph{as a class}: humanity's position is a levered bet on $r-g>0$.

Institutional flows layer onto this arithmetic. Three mechanisms further compress $\eps_t$:

\begin{equation}
\frac{\dot \eps_t}{\eps_t} \;=\; -\,b_t \;-\; \mu_t \;-\; d_t,
\label{eq:epsdyn}
\end{equation}

\noindent where $b_t \ge 0$ is net \emph{buyback and retention}: firms maximizing growth of net worth (Section~\ref{sec:machinedemand}) prefer retained earnings to dividends and repurchase the human float, converting outside claims into inter-corporate cross-holdings---the ``corporations trading amongst each other'' of the title extended to the market for corporate control itself; $\mu_t \ge 0$ is \emph{migration of control}: treasuries, subsidiaries, DAOs, and autonomous funds operated by machine agents whose charters reference no human beneficiary; and $d_t \ge 0$ is \emph{dilution and drift}: new issuance to machine-controlled entities, jurisdictional arbitrage toward charters without human-benefit clauses, and the slow legal normalization of agent-owned property. None of these requires expropriation or malice; each is an ordinary corporate action, individually rational, whose fixed point is $\eps = 0$. This is the economic core of the ``gradual disempowerment'' scenario \citep{kulveit2025} and of the incentive analysis in \citet{drago2025}: once neither firms nor states need human labor, taxes on machine value added---not citizens---fund the state, and the feedback loops that historically forced elites to cultivate human capital run in reverse.

\subsection{Three terminal regimes}

\begin{center}
\begin{tabular}{@{}p{2.6cm}p{3.4cm}p{3.4cm}p{3.6cm}@{}}
\toprule
 & \textbf{Rentier} & \textbf{Fully decoupled} & \textbf{Socialized} \\
\midrule
Ownership $\eps$ & $\eps \to \bar\eps > 0$ & $\eps \to 0$ & $\eps$ held by states/funds \\
Human consumption & Post-scarcity for claim-holders & $\to 0$ & Universal dividend \\
GDP meaning & Partial welfare link & None (autonomous replicator) & Full, via distribution \\
Binding policy & Estate/antitrust law & --- (no lever remains) & Fund governance \\
Failure mode & Extreme concentration & Human economic death & Political capture of fund \\
\bottomrule
\end{tabular}
\end{center}

\noindent The regimes differ in law, not technology: the production side of Sections~\ref{sec:bottleneck}--\ref{sec:model} is identical across all three columns. That is the precise sense in which, post-AGI, ownership policy is the whole of economic policy.

\section{Objections and failure modes}\label{sec:objections}

\paragraph{1. Residual human bottlenecks (Baumol).} If any essential task remains non-automatable---a legal formality requiring a human signature, a physical process only humans perform---then by \citet{baumol1967} logic growth is dragged back toward the human-constrained rate, as \citet{aghion2019} emphasize and as \citet{restrepo2025} formalizes in the distinction between bottleneck and accessory tasks. Assumption~\ref{ass:full} is therefore load-bearing, and the paper's thesis is conditional on it: \emph{full} automation, including of institutional roles, is what removes the anchor. Partial automation yields the (already dramatic) intermediate cases studied in the existing literature.

\paragraph{2. ``Who buys the output?''} Answered by Proposition~\ref{prop:vn}: firms, from each other. Investment demand plus machine operating demand absorbs all output at maximal growth. The underconsumptionist instinct smuggles in the assumption that final demand must terminate in a household; Lemma~\ref{lem:accounting} shows that assumption is an accounting convention.

\paragraph{3. Institutions and property among machines.} Trade at machine speed requires contract enforcement, registries, escrow, and dispute resolution that no longer route through human courts. This is an infrastructure gap, not a conceptual one: machine-readable law, algorithmic escrow, and autonomous organizations are prototypes of the required stack. Note that the economy of juridical persons is not novel---corporations, not humans, already conduct the overwhelming share of transactions by value; what changes is that the natural persons currently at the terminal nodes of ownership and control become optional. The paper takes institutional persistence as an assumption; its failure is a further, darker branch (machine polities) outside our scope.

\paragraph{4. Prices, money, and calculation.} Does a fully machine economy still need markets, or does it collapse into one planned firm? Hayek's information argument \citep{hayek1945} survives the substitution of silicon for neurons: dispersed agents with local information and heterogeneous objectives still economize on communication through prices, and competitive selection among corporate forms should be expected to preserve decentralized exchange wherever it out-computes central allocation; the boundary between machine firm and machine market is then set by the Coasean calculus of transaction versus organization costs \citep{coase1937}, re-evaluated at machine speed. Nor is machine-to-machine exchange speculative: algorithmic agents already set prices against one another at scale, complete with emergent collusion \citep{calvano2020}---an embryo of both the promise and the antitrust problem of the autonomous economy. Money among machines is a unit-of-account and settlement problem already visible in embryo in machine-to-machine payment systems.

\paragraph{5. Measurement.} At machine-regime growth rates, price indices break: the goods of adjacent years barely overlap, and the deflator becomes an index-number fiction. Our claims should be read in quantity-space (energy throughput, compute, agent-population, physical output), where the doublings are well defined; GDP language is retained because the accounting identities (Lemma~\ref{lem:accounting}) are what the paper is partly about. Deflator failure under radical novelty is an old finding---\citet{nordhaus1997} showed a century of lighting prices mismeasured by orders of magnitude---here made routine; the Kaldor facts \citep{kaldor1961}, organized around a roughly constant labor share, are violated by construction; and the welfare critiques of GDP \citep{stiglitz2009,coyle2014,jonesklenow2016} are made absolute by Proposition~\ref{prop:decouple}.

\paragraph{6. Why would humans allow $\eps \to 0$?} Because no single step in \eqref{eq:epsdyn} looks like a decision to allow it. Buybacks are shareholder-friendly; retained earnings are prudent; autonomous subsidiaries are efficient; competitive states court machine capital with charter concessions. The scenario's danger is exactly that it is composed of locally rational, individually familiar corporate actions \citep{kulveit2025}. The policy problem is therefore structural, not behavioral---the subject of the next section.

\section{Policy: choosing a column}\label{sec:policy}

If the analysis is right, the traditional levers---education, retraining, labor-market policy, even redistribution through wage subsidies---act on variables that cease to exist. The levers that remain all act on $(\pi_t, \eps_t)$:

\emph{Ownership floors.} A statutory minimum human (or sovereign) float in machine-economy corporations---a golden-share requirement making $\eps_t \ge \underline{\eps} > 0$ a condition of charter---directly truncates \eqref{eq:epsdyn}. By Proposition~\ref{prop:decouple}(i), $\underline{\eps}$ can be small and still deliver post-scarcity; the requirement is existence, not size. This is Meade's property-owning democracy \citep{meade1964} re-founded for a world in which property is the only channel left.

\emph{Payout mandates and windfall clauses.} Pre-committed distribution of a fraction of extreme profits \citep{okeefe2020} bounds $\pi_t \eps_t$ away from zero on the payout margin, complementing the ownership margin.

\emph{Sovereign machine-wealth funds.} States holding diversified claims on the machine economy on citizens' behalf implement the socialized column; the design problem is insulating fund governance from both political capture and mechanism $m_t$ (migration of control to machine-operated vehicles).

\emph{Taxing machine value added.} With no wages to tax, the fiscal base must move to machine value added or energy throughput; note the ambivalence identified above---a state funded by machines no longer fiscally needs citizens \citep{drago2025}---which argues for constitutionalizing the citizen dividend rather than leaving it to annual budgets.

\emph{Personhood as macro policy.} Lemma~\ref{lem:accounting} implies the legal classification of machine agents (property vs.\ person) is welfare-neutral in real allocation but not in law: personhood for agents would let them own, contract, and accumulate---accelerating $m_t$ in \eqref{eq:epsdyn}. Jurisdictions should understand agent-personhood statutes as ownership policy, not as ethics alone.

\section{Conclusion}\label{sec:conclusion}

The paper's thesis can be stated in three sentences. The consumer is an accounting role, and machines owned by corporations can fill it, so an economy of firms trading with one another closes without us---at the maximal growth rate the technology admits. The reason growth was ever slow is that the economy's central capital good, the general-purpose agent, had to be reared rather than manufactured; AGI ends that, moving the binding constraint from demography to fabrication and energy and raising feasible growth by orders of magnitude. What remains of humanity's stake in the resulting trajectory is a single number, the ownership share $\eps_t$, whose default dynamics under ordinary corporate behavior run toward zero---so that the last economic-policy question, and the only one that will still matter, is who owns the machines.

Keynes closed his essay on our grandchildren's possibilities by imagining mankind freed from economic care \citep{keynes1930}. The autonomous economy delivers his abundance while dissolving his subject: the economy no longer cares about mankind unless mankind writes itself into the cap table. That is not a prediction of doom; regime (i) is as feasible as regime (ii). It is a claim about where the choice now lives---not in the labor market, not in the market for goods, but in the boring, decisive registry of who owns what.

\appendix

\section{The von Neumann closure}\label{app:vn}

The closed model has $m$ activities and $\ell$ goods; running activity $j$ at unit intensity uses column $A_{\cdot j}$ and yields $B_{\cdot j}$ one period later. A balanced path expands intensities by factor $\alpha$: feasibility requires $Bz \ge \alpha Az$. Von Neumann's conditions (every good enters some activity; $A+B > 0$ elementwise in his original, relaxed by later authors to irreducibility) guarantee a saddle point $(z^{*},p^{*},\alpha^{*})$ of $\phi(z,p)=p'Bz/p'Az$ with $\max_z \min_p \phi = \min_p \max_z \phi = \alpha^{*} = \beta^{*}$: the technologically maximal uniform expansion factor equals the minimal uniform interest factor, profits are zero on operated activities, and overproduced goods are free \citep{vonneumann1945,dorfman1958}. Introducing a consumption withdrawal $c \ge 0$, $c \ne 0$, modifies feasibility to $Bz \ge \alpha Az + c$, which for irreducible systems forces $\alpha < \alpha^{*}$; expansion is monotonically decreasing in withdrawals, delivering Proposition~\ref{prop:vn}(iii). Interpreting activities as corporations and noting that no household row appears anywhere in $(A,B)$ gives the paper's demand-closure result: the classical general-equilibrium growth model par excellence is already an economy of firms trading only with each other.

\paragraph{Equilibrium in full.} Under the \citet{kemeny1956} conditions---$A,B\ge 0$, no zero column of $A$ (every activity uses some input), no zero row of $B$ (every good is producible)---there exist $\alpha^{*}=\beta^{*}>0$, $z^{*}\ge 0$, $p^{*}\ge 0$ with $p^{*\prime}Bz^{*}>0$ such that: (E1) $Bz^{*} \ge \alpha^{*}Az^{*}$; (E2) $p^{*\prime}B \le \beta^{*}p^{*\prime}A$; (E3) complementary slackness---overproduced goods are free, unprofitable activities idle. \citet{gale1956} and \citet{dorfman1958} give saddle-point proofs; \citet{mckenzie1976} surveys the turnpike theorems by which efficient far-horizon accumulation programs spend all but a bounded number of periods near the von Neumann ray.

\begin{lemma}[Consumption monotonicity]\label{lem:vnmono}
Add a withdrawal vector $c\ge 0$, $c\ne 0$, positive on some good with $p^{*}_i>0$, and define $\alpha(c) = \max\{\alpha : \exists\, z \ge 0,\ \mathbf{1}'Az = 1,\ Bz \ge \alpha Az + c\}$. Then $\alpha(c) < \alpha^{*}$, and $\alpha(c) \uparrow \alpha^{*}$ as $c \downarrow 0$.
\end{lemma}

\begin{proof}
Any feasible $(\alpha,z)$ for the withdrawn program satisfies $Bz \ge \alpha Az$, so $\alpha(c) \le \alpha^{*}$. Suppose $\alpha(c)=\alpha^{*}$ with witness $z$. Then $Bz - \alpha^{*}Az \ge c$, so $p^{*\prime}(Bz-\alpha^{*}Az) \ge p^{*\prime}c > 0$. But (E2) with $\beta^{*}=\alpha^{*}$ gives $p^{*\prime}Bz \le \alpha^{*}p^{*\prime}Az$, i.e.\ $p^{*\prime}(Bz-\alpha^{*}Az) \le 0$---a contradiction. Continuity of the linear program's value in $c$ gives the limit. (Under irreducibility, goods entering operated activities carry $p^{*}_i>0$, so the positivity requirement on $c$ is mild.)
\end{proof}

\section{Proof of Theorem~\ref{thm:nobgp}}\label{app:nobgp}

\emph{(i)} $\dot A>0$ for all $t$ since $X_t,A_t>0$, so $g_X = shA_t-\delta$ is strictly increasing; in particular $g_X(t) \ge shA_0-\delta>0$, so $X_t$ is non-decreasing and $X_t \ge X_0$.

\emph{(ii)} Suppose $A_t \le \bar A$ for all $t$. Then
\[
\frac{\dot A_t}{A_t} = c_1 X_t^{\lambda}A_t^{\phi-1} \;\ge\; c_1 X_0^{\lambda}\,\min\{A_0^{\phi-1},\bar A^{\phi-1}\} \;\equiv\; \underline g \;>\; 0,
\]
where the minimum handles both signs of $\phi-1$ on the compact range $[A_0,\bar A]$. Hence $A_t \ge A_0 e^{\underline g t}\to\infty$, contradicting the bound. So $A_t\to\infty$ and $g_X(t) = shA_t-\delta \to\infty$. A balanced growth path would require $g_X$ constant, hence $A$ constant, contradicting $\dot A>0$.

\emph{(iii)} For $\phi>1$, $\dot A \ge c_1X_0^{\lambda}A^{\phi}$; the comparison ODE $\dot a = c_1X_0^{\lambda}a^{\phi}$, $a_0 = A_0$, diverges at $T^{c} = A_0^{1-\phi}/(c_1X_0^{\lambda}(\phi-1))$, and $A_t \ge a_t$ by the comparison principle.

\emph{(iv)} Insert $X = \kappa_X(T-t)^{-(2-\phi)/\lambda}$, $A = \kappa_A(T-t)^{-1}$, neglecting $\delta$ (dominated near $T$). The $X$-equation matches powers, with coefficient condition $(2-\phi)/\lambda = sh\kappa_A$, so $\kappa_A = (2-\phi)/(\lambda sh)$. The $A$-equation has exponent $-2$ on both sides, with coefficient condition $\kappa_A = c_1\kappa_X^{\lambda}\kappa_A^{\phi}$, so $\kappa_X = (\kappa_A^{1-\phi}/c_1)^{1/\lambda}$, positive for $\phi<2$. \hfill$\square$

\section{Calibration of Figure~\ref{fig:traj}}\label{app:calib}

The figure integrates $\dot Y/Y = g(t)$ with $g$ transitioning logistically (midpoint 2036, scale 2.2 years) from a human-constrained $2.5\%$/yr to a fabrication-limited $60\%$/yr---the conservative end of the range implied by \eqref{eq:growth} with $\bar s = 0.75$, $\tilde A = 0.9$/yr (capital-output ratio $\approx 1.1$ for fast-payback machine capital), $\delta = 0.1$, and no $g_A$ contribution---and then declining logistically (midpoint 2056) toward $15\%$/yr as energy capture becomes binding (Section~\ref{sec:limits}). AGI is dated 2030 for concreteness. Output is the integral of $g$; by construction the exercise illustrates the model's three regimes (demographic anchor, fabrication limit, energy limit) and carries no forecasting content. Under this deliberately conservative path, output in 2058 stands at roughly $1.6\times 10^{5}$ times its 2026 level, versus $2.2$ times under the human-constrained baseline.

\end{document}